\documentclass[letterpaper, 10 pt, conference]{ieeeconf}  % Comment this line out if you need a4paper

\usepackage[utf8]{inputenc}
\usepackage{graphicx}
\usepackage{amsmath,systeme,amssymb,amsfonts}
\usepackage[version=4]{mhchem}
\usepackage{siunitx}
\usepackage{longtable,tabularx}
\usepackage{comment}
\usepackage{float}
\def\vec{{\rm vec}}
\usepackage{amsmath,bm}
\usepackage[english]{babel}
\usepackage{amssymb}
\usepackage[utf8]{inputenc}
\usepackage{textcomp}
\usepackage{graphicx}

\usepackage{amsmath}
\usepackage[version=4]{mhchem}
\usepackage{siunitx}
\usepackage{esvect}
\usepackage{longtable,tabularx}
 \usepackage{cite}
\usepackage{thmtools}
\usepackage{textcomp}
\usepackage{soul}
\usepackage{amssymb}
\usepackage{graphics} % for pdf, bitmapped graphics files
\usepackage{booktabs} % For formal tables
\usepackage{caption, subcaption}
\usepackage{breqn}
\usepackage{array, mathtools}
\usepackage{siunitx}
\usepackage{algorithm}
\usepackage[noend]{algpseudocode}
\usepackage{tabulary}
\usepackage{stmaryrd} % for double brackets
\usepackage{varwidth}
\usepackage{comment}
\newtheorem{theorem}{Theorem}
\newtheorem{definition}{Definition}

\newtheorem{remark}{Remark}

\newtheorem{example}{Example}
\newcommand{\D}{\mathcal{D}}
\newcommand{\R}{\mathbb{R}}

\renewcommand{\S}{\mathcal{S}}

\newcommand{\X}{\mathcal{X}}

\newcommand{\A}{\mathcal{A}}

\newcommand\sbullet[1][.5]

\usepackage{amsmath,scalerel}

\DeclareMathOperator*{\Bigcdot}{\scalerel*{\cdot}{\bigodot}}

\DeclareMathOperator*{\argmin}{arg\, min} % thin space, limits underneath in displays

\DeclareMathOperator*{\argmax}{arg\, max} % thin space, limits underneath in displays

\IEEEoverridecommandlockouts                          
\title{\LARGE \bf
Polynomial Lyapunov Functions and Invariant Sets from a New Hierarchy of Quadratic Lyapunov Functions for LTV Systems}

\author{Hassan Abdelraouf$^1$,  \thanks{$^1$ PhD candidate, Aerospace Engineering, University of Illinois at Urbana-Champaign, Illinois, USA, hassana4@illinois.edu.}
       Eric Feron$^2$, \thanks {$^2$ Professor, Department of  Electrical Engineering, King Abdullah University of Science and Technology, Thuwal, KSA,  eric.feron@kaust.edu.sa.} and
      Jeff S. Shamma$^3$ \thanks{$^3$ Professor, Industrial and Enterprise Systems Engineering, Univresity of Illinois at Urbana-Champaign, Illinois, USA,  jshamma@illinois.edu}}

\begin{document}

\maketitle
\thispagestyle{empty}
\pagestyle{empty}

%%%%%%%%%%%%%%%%%%%%%%%%%%%%%%%%%%%%%%%%%%%%%%%%%%%%%%%%%%%%%%%%%%%%%%%%%%%%%%%%
\begin{abstract}
  We introduce a new class of quadratic functions based on a hierarchy of linear time-varying (LTV) dynamical systems. These quadratic functions in the higher order space can be also seen as a non-homogeneous polynomial Lyapunov functions for the original system, i.e the first system in the hierarchy. These non-homogeneous polynomials are used to obtain accurate outer approximation for the reachable set given the initial condition and less conservative bounds for the impulse response peak of linear, possibly time-varying systems. In addition, we pose an extension to the presented approach to construct invariant sets that are not necessarily Lyapunov functions. The introduced methods are based on elementary linear systems theory and offer very much flexibility in defining arbitrary polynomial Lyapunov functions and invariant sets for LTV systems. 
\end{abstract}

%%%%%%%%%%%%%%%%%%%%%%%%%%%%%%%%%%%%%%%%%%%%%%%%%%%%%%%%%%%%%%%%%%%%%%%%%%%%%%%%
\section{INTRODUCTION}
Linear time-varying (LTV) systems appear in many applications such as changing the aerodynamic coefficients of the aircraft with the flight speed and altitude or changing the parameters of chemical plants or electrical circuits. In control theory, it is common to treat some nonlinear systems as LTV systems and linear switching systems by 
 local linearization around a set of operating points \cite{boyd1994linear}. Lyapunov functions are a  widely used concept to tackle stability and related analyses for linear systems in control literature. For example, the stability of a linear time-invariant (LTI) system is equivalent to the existence of  a quadratic Lyapunov function. However, the stability of LTV systems does not necessarily imply the existence of such a quadratic Lyapunov function, but it is equivalent to the existence of  a polynomial homogeneous higher order Lyapunov function\cite{ahmadi2017sum,mason2006common}. 
 
 Stability and control of switching systems have been widely studied \cite{liberzon2003switching}. In \cite{zelentsovsky1994nonquadratic}, polynomial, homogeneous Lyapunov functions for uncertain systems is shown to be equivalent to a quadratic Lyapunov function for a transformed system. Authors in \cite{abate2020lyapunov}, generalized this idea and defined a hierarchy of dynamical systems using a lifting procedure such that the quadratic Lyapunov functions for any system in the hierarchy can be used as a polynomial homogeneous Lyapunov function for the original system. Then, the search for homogeneous Lyapunov function for LTV systems can be recast as the search for a quadratic Lyapunov function in a related hierarchy. In this paper, we extend the contents of \cite{abate2020lyapunov} and add more flexibility to the methods described therein to  consider polynomial, but not necessarily homogeneous  Lyapunov functions as well. This extension is powerful in producing better analysis metrics for LTV systems such as less conservative outer approximation for the reachable set of states given the initial condition than the results introduced in \cite{abate2020lyapunov}.  
 An alternative approach  in computing polynomial Lyapunov functions and various performance and robustness guarantees for switching linear system is Sum-of-Squares (SOS) optimization \cite{parrilo2003semidefinite} in control systems \cite{jarvis2005control}. SOS techniques cast the search for polynomial Lyapunov  or Lyapunov-like functions as a convex feasibility problem for which many solvers were developed to solve \cite{parrilo2000structured}. 
 
 Lyapunov-like functions are used to capture point-wise-in-time metrics for LTI systems such as peak norms. But, these metrics are captured with high conservatism \cite{abedor1996linear, blanchini1995nonquadratic}. However,  polynomial homogeneous Lyapunov-like functions are utilized to reduce this conservatism and get more accurate upper bounds \cite{abate2021pointwise,abdelraouf2022computing}. In this work, we use polynomial non-homogeneous Lyapunov-like functions to get more accurate upper bounds for the peak of the impulse response of LTV systems, moreover, these non-homogeneous polynomials are used to generate a worst case trajectory to accurately bound the peak norm from below as well. 
The work in this paper is motivated by the fact that someone with basic knowledge about linear systems theory and convex optimization tools can easily build the system hierarchy and produce useful performance analysis for LTV systems. Then, the  contribution can be summarized by (1) introducing a new hierarchy of LTV systems where the quadratic Lyapunov function for the lifted system in any level can be used as a non-homogeneous Lyapunov function for the original system. (2) Using these non-homogeneous Lyapunov functions to get better approximations for the reachable sets compared to quadratic functions in \cite{boyd1994linear} and homogeneous polynomials in \cite{abate2020lyapunov}. (3) Using the hierarchy of LTV systems to obtain non-homogeneous polynomial invariant sets that are not necessarily Lyapunov functions.  
 
%%%%%%%%%%%%%%%%%%%%%%%%%%%%%%%%%%%%%%
\section{Notations}

Denote the set of real numbers by $\R$ and the set of non-negative real numbers by $\R^+$. Denote by $S_{++}^n\subset \R^{n\times n}$ the set of symmetric positive definite $n \times n$ matrices. For $P \in \R^{n\times n}$, $P \succ 0$ means that $P \in \S_{++}^{n \times n}$ such that the quadratic form $V(x)=x^T P x$ is positive for all nonzero $x\in \R^n$. The zero vector in $\R^n$ is denoted by $0_n \in \R^n$ and the identity $n \times n$ matrix is denoted by $I_n$. The convex hull of the set $\mathcal{M} \subset \R^{n\times n}$ is denoted by $\text{conv}(\mathcal{M}) \subset \R^{n \times n}$. The vectorization of a matrix $P\in \R^{n \times n}$, with $n$ columns denoted by $P_{\Bigcdot 1}, \dots, P_{\Bigcdot n} $ , is $\vec({P}) \in \R^{n^2}$ such that $\text{vec}(P)=\begin{bmatrix} P_{\Bigcdot 1}^T & \dots &P_{\Bigcdot n}^T\end{bmatrix}^T$. For a set of matrices $A_i \in \R^{n_i\times n_i }$ for $i=1,\dots,p$, we define $\text{diag}(A_1,A_2, \dots,A_p)$ as 
\begin{equation}
    \text{diag}(A_1,A_2, \dots A_p)= \left[ \begin{array}{cccc} 
{A}_1 & 0 &\ldots & 0 \\
0 & {A}_2 & \ddots & \vdots \\
\vdots & \ddots & \ddots & 0 \\
0 & \ldots & 0 & {A}_p \\
\end{array}\right].
\end{equation}
 The \textit{Kronecker product} of $A\in \R^{n \times m}$ and $B \in \R^{p\times q}$ is denoted by $A \otimes B \in \R^{n p \times mq}$ which is defined as
\begin{equation}
    A \otimes B := 
    \begin{bmatrix}
    a_{11} B & \cdots & a_{1m} B\\
    \vdots & \ddots & \vdots\\
    a_{n1} B & \cdots & a_{nm} B
    \end{bmatrix},
\end{equation}
where $a_{ij}$ is the $(i,j)$-th entry of $A$ for $i=1,\dots,n$ and $j=1,\dots,m$. $(A\otimes B)(C \otimes D)=AC \otimes BD$ for $A,B,C$ and $D$ with proper dimensions is an important property of Kronecker product that is used in this note. Moreover, $(A \otimes B)^T= A^T \otimes B^T$. For an integer $i\geq2$, the $i^{th}$ Kronecker power of a matrix $A \in \R^{n \times m}$ is denoted by $\otimes^i A$ such that $\otimes ^1 A =A$ and $\otimes^{k} A =A \otimes (\otimes^{k-1}A)$ for all $k\ge 2$.

%%%%%%%%%%%%%%%%%%%%%%%%%%%%%%%%%%%%%%%%%%%%%%%%%%%%
\section{stability of Linear-Time Varying systems}

Consider the linear time-varying system 
\begin{equation}
\dot{x}=A(t)x,\;\; x(0) = x_0,
\label{ LTV system}
\end{equation}
where $x \in \R^{n}$ is the system state vector and $A(t)$ evolves inside the set $\text{conv}(\mathcal{M}) \subset \R^{n\times n}$ for all $t \geq 0$ and $\mathcal{M}=\{A_1,A_2,\dots,A_N\}$. We assume that each mode $\dot{x}=A_i x$ for $i=1,\dots, N$ is asymptotically stable, i.e. the system states converge asymptotically to the origin starting from any initial condition $x(0)$. It is shown in \cite{liberzon2003switching} that the asymptotic stability of each switching mode does not necessarily imply the stability of the system (\ref{ LTV system}) under arbitrary switching. 
\begin{remark}
    System (\ref{ LTV system}) defines a class of LTV systems that are called linear switched systems under arbitrary switching as in \cite{lin2009stability} or linear differential inclusions as in \cite{boyd1994linear}. Moreover, the time varying polytopic systems defined in \cite[chapter 3]{chesi2009homogeneous} can be written in form (\ref{ LTV system}). For example, consider the system 
    \begin{equation*}
    \dot{x}=A_0+p_1(t)A_1+p_2(t)A_2 
    \end{equation*}
    where $A_0,A_1,A_2 \in \R^{n\times n}$ are given matrices and the uncertain parameters $p_1(t),p_2(t) \in [-1,1]$. This system can be written as a LTV system in the form (\ref{ LTV system}) for $A(t) \in \text{conv}\left({\mathcal{M}}\right)$, where
    \begin{align*}
        \mathcal{M}=\{&A_0 + A_1 +A_2 ,\\ &A_0+A_1-A_2,\\
        &A_0-A_1+A_2,\\
        &A_0-A_1-A_2\}.
    \end{align*}
\end{remark}

 In the context of Lyapunov stability theory, the system (\ref{ LTV system}) is globally asymptotically stable if there exists a radially unbounded function $V:\R^n \to \R$ such that 
 \begin{equation}
 \begin{aligned}
 V(0)&=0\\
 V(x)&>0 \; \;  \text{ for all nonzero } x  \\ 
 \dot{V}(x)&= \langle \frac{\partial V}{\partial x}, A_j x \rangle < 0 \; \; \text{for all } j\in \{1,2, \dots, N\}.
 \end{aligned}
 \label{Lyap condtions} 
 \end{equation}
 The system (\ref{ LTV system}) is said to be quadratically stable when the Lyapunov function is quadratic in $x$ such that $V(x)=x^T P x$ where $P \in \S_{++}^n$. Then,  the negative gradient condition in (\ref{Lyap condtions}) becomes
 \begin{equation}
 PA_j+A_j^T P \prec 0  \; \; \text{for all } j\in \{1,2, \dots, N\}. 
 \label{Lyap quadratic conditions}
 \end{equation}
 The search for quadratic Lyapunov functions for (\ref{ LTV system}) has some computational advantages compared to other stability analysis methods \cite{vandenberghe1993polynomial}. The Linear Matrix Inequality (LMI) (\ref{Lyap quadratic conditions}) can be solved efficiently by many available semi-definite program solvers. If the system (\ref{ LTV system}) is time invariant i.e. $N=1$, then the stability is equivalent to the existence of a quadratic Lyapunov function, however, in the general setting $N\geq 2$, the system (\ref{ LTV system}) can be stable, yet no quadratic Lyapunov function certifies its stability. This motivates the search for higher order polynomial Lyapunov functions. In \cite{abate2020lyapunov}, the authors provided tools to compute higher order, polynomial and homogeneous Lyapunov functions for (\ref{ LTV system}). The \textit{Lyapunov differential equation} for system (\ref{ LTV system}) is 
 \begin{equation}
 \dot{X}=A(t)X+XA(t)^{T},
 \label{Lyap diff eq}
 \end{equation}
 where $X \in \R^{n \times n}$. The differential equation (\ref{Lyap diff eq}) can be written as 
 \begin{equation}
 \dot{\vv{X}}=\A(t) \vv{X}
 \label{vect lyap equation}
 \end{equation}
 where $\vv{X}=\text{vec}(X)$. In this case,  $\A(t) \in \R^{n^2 \times n^2}$ evolves in $\text{conv}(\mathcal{M})$ with $\mathcal{M}=\{\A_1,\A_2,\dots,\A_N\}$ and
 \begin{equation}
 \A_j=I_n \otimes A_j+A_j \otimes I_n \; \; \text{for all } j=1,\dots,N. 
 \label{lyap eq kron}
 \end{equation}
 From several prior works, it is shown that the system (\ref{vect lyap equation}) is stable if and only if system (\ref{ LTV system}) stable. Moreover, the formula (\ref{lyap eq kron}) can be generalized to produce a hierarchy of dynamical systems with higher dimensions. Indeed, the system (\ref{ LTV system}) is considered to be the first system, denoted $H_1$, in the hierarchy obtained by the recursion 
 \begin{equation}
  H_i:  \left\{
\begin{array}{rcl} 
\dot{\xi}_i &=& \A^i(t)\xi_i\\
\A^i(t) &\in& \text{conv}(\mathcal{M}_i) \\
\mathcal{M}_i &=&  \{\A_1^i,\, \cdots,\, \A_N^i\}\\
\A^i_j &=& I_n \otimes \A_j^{i-1} + A_j\otimes I_{n^{i-1}}
\end{array}\right.
\label{systems Hierarcy}
 \end{equation}
where, $i\geq2$, $\xi_1=x$ and $\xi_i=\otimes^i x$, so $\xi_i \in \R^{n^i}$. In \cite{abate2020lyapunov}, it is shown that the existence of a quadratic Lyapunov function for any system $H_i$ in the hierarchy (\ref{systems Hierarcy}) implies the existence of a polynomial homogeneous Lyapunov function of order $2i$ for the system (\ref{ LTV system}). This leads the authors to develop a series of applications of such homogeneous Lyapunov functions for the analysis of LTV systems such as  increasingly tight approximations of time domain input-output stability metrics and worst case system time response.  In the next section, we complement this work by introducing a new hierarchy of LTV systems such that a quadratic Lyapunov function for any system in the hierarchy can be considered as a non-homogeneous polynomial Lyapunov function for the original system. 

%%%%%%%%%%%%%%%%%%%%%%%%%%%%%%%%%%%%%%%%%%%%%%%%%%%%
 \section{A new hierarchy of dynamical systems and corresponding quadratic Lyapunov functions}
\label{new hierarchy}
 
 Motivated by the search for polynomial Lyapunov and Lyapunov-like functions for system (\ref{ LTV system}) that are not necessarily homogeneous, we introduce a new hierarchy of dynamical systems where system (\ref{ LTV system}) is the first system in the hierarchy and for $i\ge 2$,  
 
 \begin{equation}
  \tilde{H}_i:  \left\{
\begin{array}{rcl} 
\dot{\tilde\xi}_i &=& \tilde \A^i(t)\tilde \xi_i\\
\tilde \A^i(t) &\in& \text{conv}( \tilde{ \mathcal{M}}_i) \\
\tilde{\mathcal{M}}_i &=&  \{\tilde \A_1^i,\, \cdots,\, \tilde \A_N^i\}\\
\tilde \A^i_j &=& \text{diag}(\A^1_j,\A^2_j,\dots,\A^i_j) 
\end{array}\right.
\label{ new systems Hierarcy}
 \end{equation}
where $\A^k_j = I_n \otimes \A_j^{k-1} + A_j\otimes I_{n^{k-1}} \; \; \text{for all } k=2,\dots i$ and  $\A_j^1 =A_j$ for all $j=1,\dots,N$.  From the hierarchy (\ref{systems Hierarcy}), $\xi_i=\otimes^i  x$ for $i\geq 2$ such that $\xi_1=x$. Then, in this case,  $\tilde{\xi}_i=\left[\xi_1^T \; \xi_2^T \; \ldots \; \xi_i^T  \right]^T$ with dimension 
\begin{equation*}
\tilde{n}_i=n+n^2+\dots+n^i = \frac{n(n^i-1)}{(n-1)}.
\end{equation*}
In essence, the systems defined by (\ref{ new systems Hierarcy}) are obtained by forming parallel concatenations of the systems defined by (\ref{systems Hierarcy}).  We then  consider the quadratic Lyapunov and Lyapunov-like functions for the system $\tilde{H}_i$ which is called the lifted system of degree $i$.
\begin{theorem}
    If system (\ref{ LTV system}) is quadratically stable, then for every $i\geq2$, there exists a quadratic Lyapunov function which proves the stability of $\tilde{H}_i$ and if $P_1$ satisfies (\ref{Lyap quadratic conditions}), then $\mathcal{P}_i=\text{diag}(P_1,\otimes^2 P_1,\dots, \otimes^i P_1 )$ satisfies
    \begin{equation}
        \mathcal{P}_i \tilde{\A}^i_j+ (\tilde{\A}^i_j)^T \mathcal{P}_i \preceq 0 \; \; \text{for all } j=1,\dots,N 
        \label{quad stability for H_i}
    \end{equation}
    for the system $\tilde{H}_i$.
\end{theorem} 
\begin{proof}
Without loss of generality, let $\mathcal{M}_1 = \{A\}$. First, we prove that every subsystem $\dot{\xi}_k=\A_k \xi_k$ for all $k=1,\dots, i$ in $\tilde{H}_i$ is quadratically stable  such that $P_k \A^k + (\A^k)^T P_k \preceq 0$ and $P_k=\otimes^k P_1$. For $k=1$, it is given that the subsystem $\dot{x}=Ax$ is quadratically stable such that $P_1A+A^T P_1 \preceq 0 $. Then for $k \ge 2 $,
\begin{align*}
    P_k \A^k &= (P_1 \otimes P_{k-1})(I_n \otimes \A^{k-1} + A \otimes I_{n^{k-1}})\\
    &= P_1 \otimes P_{k-1} \A^{k-1}+P_1A \otimes P_{k-1}, \\
    (\A^k)^T P_k &= (I_n \otimes \A^{k-1} + A \otimes I_{n^{k-1}})^T(P_1 \otimes P_{k-1})\\
    &= P_1 \otimes (\A^{k-1})^T P_{k-1}+ A^T P_1 \otimes P_{k-1}.
\end{align*}
Therefore, 
\begin{equation}
    \begin{aligned}
    P_k \A^k &+ (\A^k)^T P_k = \\
    & P_1 \otimes (P_{k-1} \A^{k-1}+(\A^{k-1})^T P_{k-1})\\
    &+ (P_1 A + A^T P_1) \otimes P_{k-1}
\end{aligned}
\label{lyap for subsystem}
\end{equation}
From Kroncker product properties, for two matrices $L \in \R^{l\times l}$ with eigenvalues $\lambda_p, p =1,\dots,l$ and $M \in \R^{m\times m}$, with eigenvalues $\mu_q, q =1,\dots,m$, the eigenvalues of $L\otimes M$ is $\lambda_p \mu_q \; \text{for } p =1,\dots,l \text{ and } q =1,\dots,m$. Hence, if $L$ is positive-(semi)definite and $M$ is negative-(semi)definite, then $L\otimes M$ is negative-(semi)definite. By applying this property in (\ref{lyap for subsystem}), $(P_1 A + A^T P_1) \otimes P_{k-1} \preceq 0$ since $P_{k-1} \succ 0$. Also, By induction, the term $ P_1 \otimes (P_{k-1} \A^{k-1}+(\A^{k-1})^T P_{k-1}) \preceq 0$. Therefore,  $P_k \A^k + (\A^k)^T P_k \preceq 0$, which proves that the subsystem $\dot{\xi}_k=\A^k \xi_k$ for all $k=2,\dots, i $ is quadratically stable. Then,
\begin{equation*}
    \mathcal{P}_i \tilde{A}^i+ (\tilde{A}^i)^T \mathcal{P}_i = \sum_{k=1}^{i} P_k \A^k + (\A^k)^T P_k \preceq 0 
\end{equation*}
Hence the system $\tilde{H}_i$ is quadratically stable with $\mathcal{P}_i=\text{diag}(P_1,\otimes^2 P_1,\dots, \otimes^i P_1 )$ that satisfies (\ref{quad stability for H_i}). 
\end{proof}

 Stability of (\ref{ LTV system}) can be certified using the hierarchy (\ref{ new systems Hierarcy}) by allowing the early symmetric blocks of $\mathcal{P}_i$ to be zeros. Moreover, the value of relying on common Lyapunov and Lyapunov-like functions becomes evident once additional information is demanded from the system (\ref{ LTV system}). Consider for example, capturing the reachable set of system (\ref{ LTV system}) from a given initial condition $x_0$. Since the time varying system is nondeterministic, many possible trajectories can be generated from a given initial condition. So, we are interested in bounding the set of all possible trajectories. Based on the results of \cite{abate2020lyapunov}, approximating such a reachable set using quadratic or homogeneous polynomial Lyapunov functions is conservative. So,  hierarchy (\ref{ new systems Hierarcy}) can be used to find a polynomial and non-homogeneous approximation by solving the following semi-definite program for $i\geq1$
\begin{equation}
\begin{aligned}
    &\text{minimize } \tilde{\xi}_{i,0}^T \mathcal{P}_i \tilde{\xi}_{i,0}, \\
    &\text{subject to (\ref{quad stability for H_i}) and } \mathcal{P}_i \succeq 0
    \label{reachable set optimization}
\end{aligned}
\end{equation}
where $\tilde{\xi}_{i,0}=\left[x_0^T \; (\otimes^2 x_0)^T \; \dots \;(\otimes^i x_0)^T \right]^T $. Then, the reachable set from a given initial condition $x_0$ is 
\begin{equation*}
    \X_i = \{x: {\xi}_{i}^T \mathcal{P}_i {\xi}_{i}\leq \tilde{\xi}_{i,0}^T \mathcal{P}_i \tilde{\xi}_{i,0}\}. 
\end{equation*}
where $\tilde{\xi}_i=\left[ x^T \; (\otimes^2 x)^T\; \dots \; (\otimes^i x)^T  \right]^T$. The following example, originally used in \cite{abate2020lyapunov}, illustrates the power of this approach. 
%%%%%%%%%%%%%%%%%%%%%%%%%%%%%%%%%%%%%%%%%%%%%%%%%%%%%%
\begin{example}
Consider the linear time varying system (\ref{ LTV system}) with $\mathcal{M}=\{A_1,A_2\}$ and
\begin{equation}
A_1 = \begin{bmatrix}
    -0.5 & 0.5 \\ -0.5 &-0.5
\end{bmatrix}, \; \; A_2 =\begin{bmatrix}
    -2.5 &2.5 \\ -2.5 & 1.5
\end{bmatrix}.
\label{LTV system parameters}
\end{equation}
\end{example}
The system is known to be quadratically stable. Consider the initial condtion $x_0= \left[1\; 0\right ]^T $. The semi-definite program (\ref{reachable set optimization}) is computed using SDPT3 solver supported by CVX \cite{grant2014cvx}. The approximation of the set of reachable states starting from $x_0$ for $i=2,4, \text{and } 6$ are plotted in Figure \ref{fig:1}.
\begin{figure}[H]
    \centering
    \includegraphics[scale=0.38]{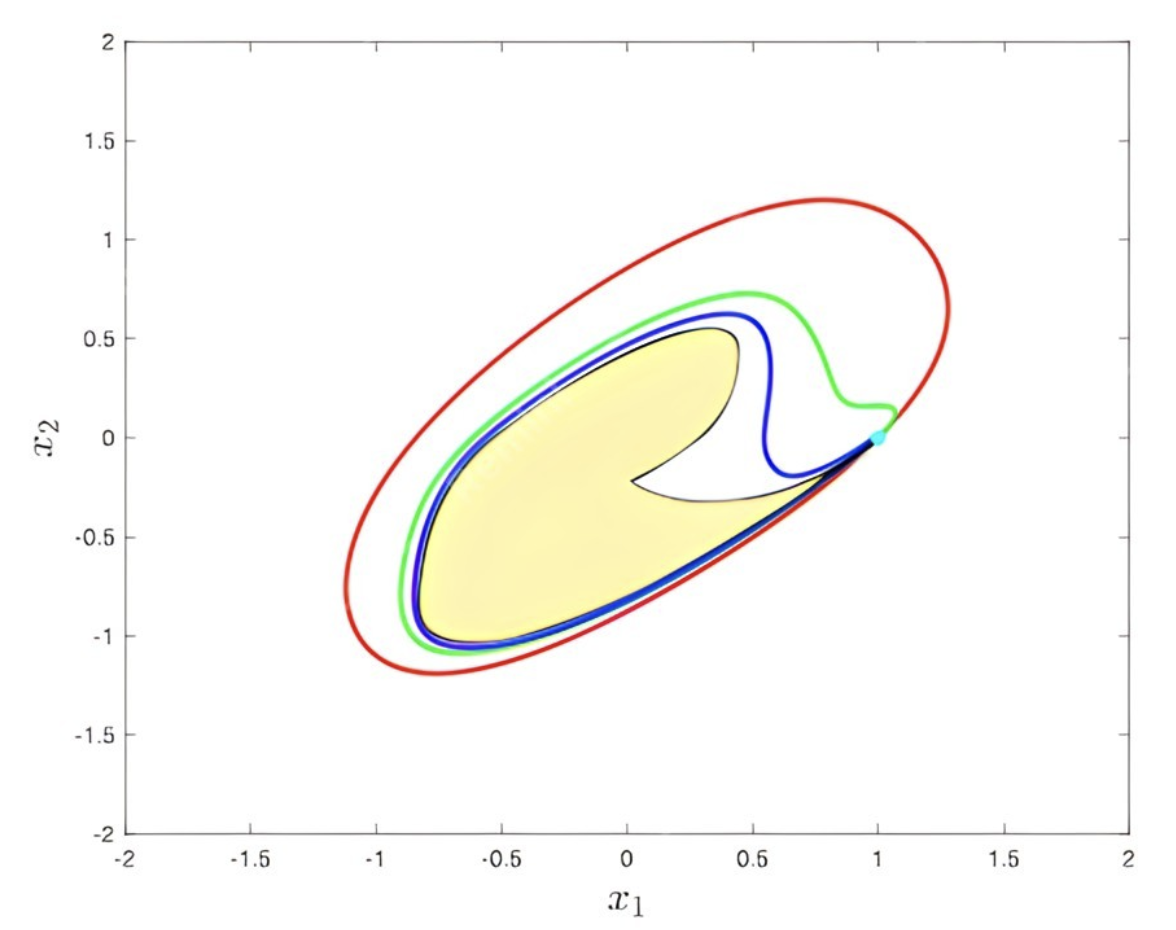}
    \caption{Simulated response of system (\ref{ LTV system}) with parameters (\ref{LTV system parameters}). The light yellow region represents the set of reachable states from $x_0$, the cyan point. The red, green and blue regions represent the approximation of the reachable set using $4^{\text{th}}$,  $8^{\text{th}}$ and $12^{\text{th}}$ order non-homogeneous polynomials computed by solving (\ref{reachable set optimization}) respectively.}  
    \label{fig:1}
\end{figure}

\begin{remark}
    The reachable sets $\mathcal{X}_i$ produced by solving (\ref{reachable set optimization}) are invariant which means that for any trajectory $x$ which starts at $x_0 \in \X_i$ under the dynamics (\ref{ LTV system}), $x(t)\in \X_i$ for all $t>0$. 
\end{remark}

From figure 1, the set of reachable states from a given initial condition is not symmetric around the origin, so finding the outer approximation for this set by means of non-homogeneous polynomials gives more accurate approximations than homogeneous polynomials that are symmetric around the origin. That makes the approximations presented here far better than those introduced in \cite{abate2020lyapunov}. Motivated by the excellent performance obtained when considering example 1, we then examine the kind of reachable set approximation that may be obtained for a LTI system running from a given initial condition. That can be illustrated by the  following example.
\begin{example}
Consider the linear time invariant (LTI) system 
\begin{equation}
\dot{x}=\begin{bmatrix}
0 & 1 \\ -2 &-1
\end{bmatrix}x, \; \;  x_0 = \begin{bmatrix} 1 \\ 0 \end{bmatrix}.
\label{LTI system}
\end{equation}
\end{example}
The semi-definite program (\ref{reachable set optimization}) is solved to bound the system trajectory by a non-homogeneous polynomial for $i=2,3, \text{and } 4$ that generates  4-th, 6-th and 8-th non-homogenuous order polynomials respectively. Figure 2 illustrates that increasing the order of the non-homogeneous polynomial produces increasingly accurate invariant polynomial outer approximations for that system's specific trajectory.  
\begin{figure}[H]
    \centering
    \includegraphics[scale=0.21]{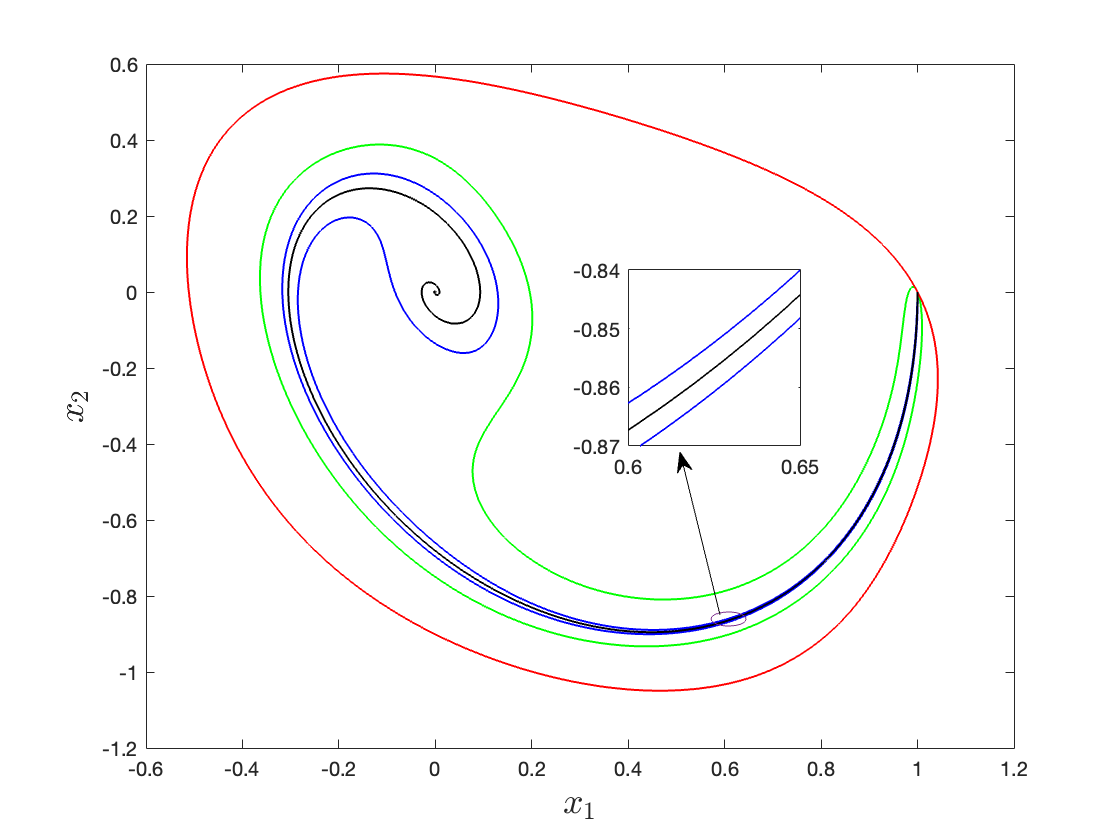}
    \caption{The black line represents the  simulated response of system (\ref{LTI system}) . The red, green and blue regions represent the approximation of the reachable set using $4^{\text{th}}$,  $6^{\text{th}}$ and $8^{\text{th}}$ non-homgenuous polynomials computed by solving (\ref{reachable set optimization}) respectively.} 
    \label{fig:2}
\end{figure}
In the next section, the introduced hierarchy of non-homogeneous Lyapunov function is used to compute less conservative bounds for the impulse response of both LTI systems and LTV  systems. 
%%%%%%%%%%%%%%%%%%%%%%%%%%%%%%%%%%%%%%%%%%%%%%%%%
\section{Analysis of impulse response via non-homogeneous Lyapunov functions}
In this section, we consider the single-input, single-output LTV system 
\begin{equation}
    \dot{x}=A(t)x + bu, \; \; y = cx, 
    \label{LTV system input}
\end{equation}
where $u$ and $y$ are the system's input and output respectively. The impulse response $h(t)$ of the system (\ref{LTV system input}) is given by the output of the system 
\begin{equation}
\dot{z}(t)=A(t)z(t), \; \;  h(t)= c z(t) ,
\label{impulse system}
\end{equation}
with initial condition $z(0)=b$. The objective is reducing the conservatism of the bounds on the peak of the impulse response for system (\ref{LTV system input}) generated using quadratic Lyapunov functions. This conservatism has been quantified in [1, page 98] or \cite{feronThesis}. The augmented system (\ref{impulse system}) allows us to take advantage of the hierarchy (\ref{ new systems Hierarcy}). Then, system (\ref{impulse system}) can be lifted to level $i\geq 2$ as 
\begin{equation}
\dot{\tilde{\xi}}_i=\tilde{\A}^i(t)\tilde{\xi}_i, \; \;  {\textbf{h}}_i(t) = \tilde{\bm{c}}_i \tilde{\xi}_i
\label{lifted impulse response}
\end{equation}
where $\tilde{\xi}_i(0)=\tilde{\bm{b}}_i=\left[b^T\;  (\otimes^2 b)^T\; \dots \; (\otimes^i b)^T\right]^T$. and $\tilde{\bm{c}}_i= \left[ c \; \otimes^2c \; \dots \otimes^i c\right]$. In this case, $\textbf{h}_i(t)=h(t)+h(t)^2+\dots+h(t)^i$. 
\begin{theorem}
   The impulse response of (\ref{LTV system input}) satisfies $|h(t)| \leq \bar{{h}}_i$, where $\bar{h}_i$ is the unique real positive root of the polynomial $p^i+p^{i-1}+\dots+p-\bar{\textbf{h}}_i$ where 
   \begin{equation}
\bar{\bm{h}}_i=\sqrt{\tilde{\bm{c}}_i \mathcal{P}_i \tilde{\bm{c}}_i^T}   \sqrt{\tilde{\bm{b}}_i^T \mathcal{P}_i \tilde{\bm{b}}_i}
\label{h bar equation}
   \end{equation}
   for any $\mathcal{P}_i$ satisfies (\ref{quad stability for H_i}). 
\end{theorem}
\begin{proof}
    Since $\mathcal{P}_i$ satisfies (\ref{quad stability for H_i}) for the system (\ref{lifted impulse response}), the set $\X_i=\{\tilde{\xi}_i \;|\; \tilde{\xi}_i^T \mathcal{P}_i \tilde{\xi}_i\leq \tilde{\bm{b}}_i^T  \mathcal{P}_i \tilde{\bm{b}}_i\}$ is invariant. Hence, the maximum output $\bar{\bm{h}}_i$ is the solution of the optimization problem 
    \begin{equation}
    \begin{array}{rcl}
    \bar{\bm{h}}_i = &  \max\limits_{\tilde{\xi}_i} & \tilde{\bm{c}}_i \tilde{\xi}_i, \\
    &\mbox{s.t.}   &\tilde{\xi}_i^T \mathcal{P}_i \tilde{\xi}_i\leq \tilde{\bm{b}}_i^T  \mathcal{P}_i \tilde{\bm{b}}_i.
    \label{impulse response optimization}
\end{array}
    \end{equation}
By applying KKT conditions, (\ref{impulse response optimization})  can be solved by $\bar{\bm{h}}_i=\sqrt{\tilde{\bm{c}}_i \mathcal{P}_i \tilde{\bm{c}}_i^T}   \sqrt{\tilde{\bm{b}}_i^T \mathcal{P}_i \tilde{\bm{b}}_i}$. which means that $|h(t)+h(t)^2+\dots+h(t)^i|\leq \bar{\bm{h}}_i$, then, $|h(t)|+|h(t)|^2+\dots+|h(t)|^i\leq \bar{\bm{h}}_i$. Hence, $|h(t)|\leq \bar{h}_i$ where $\bar{h}_i$ is root of the polynomial $p^i+p^{i-1}+\dots+p-\bar{\bm{h}}_i$. This polynomial has only one real positive root for $i \geq1$ because $\bar{\bm{h}}_i>0$ and $p^i + p^{i-1}+\dots+p$ is an increasing function of $p$ for $p \geq 0$. Therefore, only one value of $p\in \R^+$ satisfies $p^i+p^{i-1}+\dots+p-\bar{\bm{h}}_i=0$ which completes the proof. 
\end{proof}

The results of theorem 2 can be improved by optimizing (\ref{h bar equation}) over all possible $\mathcal{P}_i$. This is done by solving the convex optimization problem     
\begin{equation}
    \begin{array}{rcl}
    \mathcal{P}_i = & \argmin\limits_{Q} & {\tilde{\bm{c}}_i}Q^{-1}{\tilde{\bm{c}}}^T \vspace{.01in}\\
    & \mbox{s.t.} & {\tilde{\bm{b}}_i}^T Q{\tilde{\bm{b}}_i} \leq 1\\
    &  & ({\tilde{\A}}^{i}_j)^T Q + Q\tilde{\A}^i_{j} \preceq 0.
    \end{array}
    \label{impulse response semidefinite program}
\end{equation}
In the following example, we apply theorem 2 for a system with stiff dynamics where the bounds of the impulse response obtained using quadratic Lyapunov functions are conservative \cite{feronThesis}.
\begin{example}
Consider the following  system with stiff dynamics, that is presented in \cite{abate2021pointwise},
\begin{equation}
\dot{x}=\begin{bmatrix} -1 & 0 \\ 0 &-100 \end{bmatrix}x+ \begin{bmatrix} 1 \\ 1 \end{bmatrix}u, \; \; y = \begin{bmatrix} 
1 &-2
\end{bmatrix}x.
\label{stiff system}
\end{equation}
\end{example}
\noindent The optimization problem (\ref{impulse response semidefinite program}) is solved for $i=1, 3, \text{and } 5$, then the resulted $\mathcal{P}_i$ is used to get an upper bound for the peak of the impulse response by utilizing theorem 2.  Figure \ref{fig:3} shows that using higher order non-homogeneous polynomial Lyapunov function significantly reduces the conservatism in the impulse response bound of system (\ref{stiff system}) compared to quadratic Lyapunov function. 
\begin{figure}[H]
    \centering
    \includegraphics[scale=0.21]{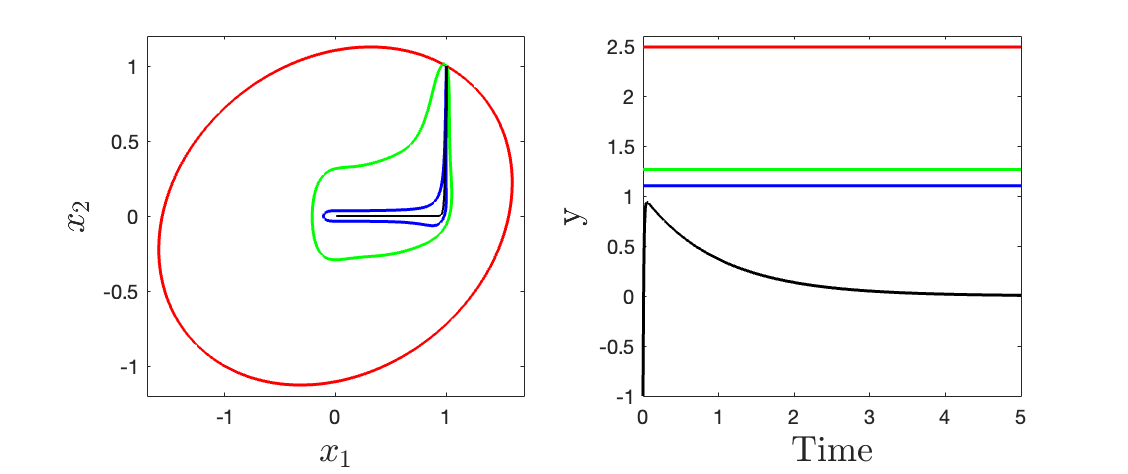}
    \caption{In the left figure , the black line represents the phase portrait of the impulse response of system (\ref{stiff system}). The red, green and blue level sets represent the invariant sets produced by solving (\ref{impulse response semidefinite program}) for $i=1,3 \text{ and } 5$ respectively. In the right figure, the black line is the impulse response of the system (\ref{stiff system}). The red, green and blue lines present the bounds obtained by theorem 2 at $i=1,3,\text{ and } 5$ respectively. } 
    \label{fig:3}
\end{figure}
In the next example, we show the effectiveness of using non-homogeneous Lyapunov functions to provide upper and lower bounds for the peak of the impulse response for uncertain systems. The upper bound is obtained by utilizing theorem 2 then Lyapunov-like function resulting from solving (\ref{impulse response semidefinite program}) is used to generate a worst case trajectory \cite{klett2021numerical}. Such a trajectory can be defined by using a Lyapunov-like function to guide the choice of $A(t)$ at every $t\ge 0$. 
\begin{definition}
    For a given continuously differentiable function $V(x)$, a worst case trajectory is defined as a trajectory $\phi(t;x_0)$ that solves (\ref{ LTV system}) with initial condition $x_0$ and $A(t)$ given by 
    \begin{equation}
A_w(t;V)= \argmax_{{A}(t) \in \textbf{conv}(\mathcal{M})} \dot{V}(\phi(t;x_0)).
\label{worst case trajectory}
\end{equation}
Intuitively, a choice of $A(t)$ that minimizes the decay rate of the Lyapunov function is "worst case". 
\end{definition}
\begin{example}
  Consider the uncertain system (\ref{LTV system input}) such that $A(t)=A+\lambda(t) \Delta$ where $\lambda \in[-1,1]$, 
  \begin{equation}
      A = \begin{bmatrix}
          0 &1 \\ -0.6 & -0.5
      \end{bmatrix}, \; \Delta=\begin{bmatrix}
          0& 0 \\ 0.1 &-0.1
      \end{bmatrix}
      \label{uncertain system parameters}
  \end{equation}
\end{example}
$b=\begin{bmatrix}
    0 & 1
\end{bmatrix}^T$ and $c=\begin{bmatrix}
    1 &0
\end{bmatrix}$. In this case, $\mathcal{M}=\{A+\Delta, A-\Delta \}$. The optimization problem (\ref{impulse response semidefinite program}) is solved for different values of $i$. Then, theorem 2 is employed to get the upper bound for  the impulse response peak. 
The resulting bounds for $i=1, 5, 10$ and $12$ are $0.9929, 0.9094, 0.8973$ and $0.8958$ respectively. 

\noindent The Lyapunov-like function $V(x)=\tilde{\xi}_i \mathcal{P}_i \tilde{\xi}_i$ is used to generate a worst case trajectory by solving (\ref{worst case trajectory}). Hence $A_w(t)=A+\lambda_{\text{max}}(t)\Delta$ such that $ \lambda_{\text{max}}(t)$ solves 
\begin{equation*}
\begin{aligned}
    \argmax_{\lambda} {\tilde{\xi}_i(t)}^T  \left[ \mathcal{P}_i (\tilde{\A}^i + \lambda \tilde{\D}^i) +  (\tilde{\A}^i + \lambda \tilde{\D}^i)^T    \mathcal{P}_i  \right] \tilde{\xi}_i(t)
\end{aligned}
    \label{lambda max optimization}
\end{equation*}
where $\tilde{\D}^i = \text{diag}(\D^1,\D^2,\dots, \D^i)$ such that $\D^1 = \Delta$ and  $\D^k =I_n \otimes \D^{k-1} + \Delta \otimes I_{n^{k-1}} $ for $k=2,\dots,i$. Therefore, $\lambda_{\max}(t)=\text{sign} ({\tilde{\xi}_i(t)}^T \left(\mathcal{P}_i \tilde{\D}_i+\tilde{\D}_i^T \mathcal{P}_i\right)\tilde{\xi}_i(t))$. The lower bound for the impulse response peak is $\max_t y_w(t)$ where $y_w$ is the output of the generated worst case trajectory. For $i=12$, the lower bound is $0.8901$ 
where the upper bound is $0.8958$ which shows the effectiveness of non-homogeneous Lyapunov function in introducing tight bounds for the peak of the impulse response of LTV systems.  

\section{Extension to invariant sets computation}
In this section, we extend our results to compute polynomial invariant sets that are not Lyapunov functions. We can use the following example to show that not every invariant set is a Lyapunov function level set and can be characterized via  $\mathcal{S}$ procedure. Consider the simple system 
\begin{equation}
\dot{x}=\begin{bmatrix} 
-1 & 0 \\0 & -1
\end{bmatrix} x 
\label{simple example invariant set}
\end{equation}
We define the set
\begin{equation*}
    S(x)=\{x \in \R^2: (x-x_0)^T(x-x_0)=3\}, \text{with } x_0=\left[ 1 \; \; 1\right]^T.
\end{equation*}
First, it is clear that $V(x)=(x-x_0)^T(x-x_0)$ is not a Lyapunov function for (\ref{simple example invariant set}) since 
\begin{align*}
    \dot{V}(x) & = -2(x-x_0)^T x\\
               & = x_0^T x_0 \; \text{when } x= x_0 \\
               & = 1. 
\end{align*}
We now show that the invariance of $V$ can be established by means of classical $\mathcal{S}$ procedure. we need to show that 
\begin{equation*}
    -(x-x_0)^T x \leq 0 \; \text{whenever } (x-x_0)^T(x-x_0)=3.  
\end{equation*}
For that purpose, it is sufficient to show that there exists a number $\lambda$ such that 
\begin{equation*}
    -(x-x_0)^Tx+\lambda \left[(x-x_0)^T(x-x_0)-3\right] \leq 0 \; \; \forall x
\end{equation*}
For $x=x_0$, $-(x-x_0)^Tx+ \lambda \left[(x-x_0)^T(x-x_0)-3\right]=-3\lambda $, thus $\lambda$ must be greater than 0. When $(x-x_0)^T(x-x_0)$ is large, then $-(x-x_0)^Tx+ \lambda \left[(x-x_0)^T(x-x_0)-3\right] = (\lambda -1)(x-x_0)^T(x-x_0)(1+o(1))$, Thus, $\lambda$ must be smaller than 1. Introduce $v=x-x_0$. Then, 
\begin{align*}
    & -(x-x_0)^Tx+\lambda \left[(x-x_0)^T(x-x_0)-3\right] \\ 
    = & -v^T v -v^T x_0 + \lambda v^T v -3 \lambda \\ 
    =& (\lambda -1) v^T v -v^T x_0 -3 \lambda 
\end{align*}
The condition of maximality of this expression is 
\begin{equation*}
    v = \frac{1}{2(\lambda-1)} x_0 
\end{equation*}
and thus the maximum of $(\lambda-1) v^T v - v^T x_0 -3 \lambda$ is 
\begin{align*}
    &\frac{-1}{4(\lambda -1)} x_0^T x_0 -3 \lambda \\
    =& \frac{-1/2-3\lambda(\lambda -1)}{\lambda -1}
\end{align*}
Choosing $\lambda=0.5$ makes this maximum negative, thus demonstrating that $S$ is invariant. Therefore, the quadratic Lyapunov function for the system hierarchy discussed in Section~\ref{new hierarchy} can be effectively generalized to polynomial invariant sets that are not necessarily Lyapunov functions by introducing the symmetric matrix 
\[
\bar{\cal P}_i = \left[ \begin{array}{cc}
{\cal P}_i & q_i \\ q_i^T & r_i\end{array} \right],
\]
for every $i\ge 1 $ in the hierarchy ,where $ \mathcal{P}_i \in \R^{\tilde{n}_i \times \tilde{n}_i}, q_i \in \R^{\tilde{n}_i}, \text{and } r_i \in \R$. We then consider the quadratic-affine function 
\[
\bar{V}(\tilde{\xi}_i)=\begin{bmatrix}
    \tilde{\xi}_i \\1 
\end{bmatrix}^T \bar{\mathcal{P}}_i \begin{bmatrix}
    \tilde{\xi}_i \\1 
\end{bmatrix}
\]
and we pose the question of the invariance of the set 
\[
\bar{\cal E}_i = \left\{ \tilde{\xi}_i \in \R^{\tilde{n}_i} \mbox{ such that } \bar{V}(\tilde{\xi}_i) \leq 0 \right\}
\]
under the action of the system $\tilde{H}_i$
in the hierarchy (\ref{ new systems Hierarcy}). Then, the invariance question can be posed as asking whether $\dot{\bar{V}} \leq 0$ whenever $\bar{V}(\tilde{\xi}_i)=0$ which can be written as 
\begin{equation*}
\begin{bmatrix}
    \tilde{\xi}_i \\ 1 
\end{bmatrix}^T \bar{\mathcal{P}}_i \begin{bmatrix}
    \tilde{\A}_j^i \tilde{\xi}_i \\ 0 
\end{bmatrix}+ \begin{bmatrix}
    \tilde{\A}_j^i \tilde{\xi}_i \\ 0 
\end{bmatrix}^T \bar{\mathcal{P}}_i
\begin{bmatrix}
    \tilde{\xi}_i \\ 1 
\end{bmatrix} \leq 0 
 \end{equation*}
for all $ j=1,\dots,N$ whenever
\begin{equation*}
    \begin{bmatrix}
    \tilde{\xi}_i \\1 
\end{bmatrix}^T \bar{\mathcal{P}}_i \begin{bmatrix}
    \tilde{\xi}_i \\1 
\end{bmatrix} =0. 
\end{equation*}
By applying $\mathcal{S}$ procedure, that holds if and only if there exists a real number $\lambda$ such that 
\begin{equation*}
    \begin{bmatrix}
    \tilde{\xi}_i \\ 1 
\end{bmatrix}^T \bar{\mathcal{P}}_i \begin{bmatrix}
    \tilde{\A}_j^i \tilde{\xi}_i \\ 0 
\end{bmatrix}+ \begin{bmatrix}
    \tilde{\A}_j^i \tilde{\xi}_i \\ 0 
\end{bmatrix}^T \bar{\mathcal{P}}_i
\begin{bmatrix}
    \tilde{\xi}_i \\ 1 
\end{bmatrix} + \lambda \begin{bmatrix}
    \tilde{\xi}_i \\1 
\end{bmatrix}^T \bar{\mathcal{P}}_i \begin{bmatrix}
    \tilde{\xi}_i \\1 
\end{bmatrix} \leq 0
\end{equation*}
for $j=1,\dots,N$ and for all $\tilde{\xi}_i \in \R^{\tilde{n}_i}$, that is, the matrix
\begin{equation}
    \begin{bmatrix}
        \bar{\mathcal{P}}_i \tilde{\A}^i_j + (\tilde{\A}^i_j)^T \bar{\mathcal{P}}_i + \lambda \bar{\mathcal{P}}_i & (\tilde{\A}_j^i)^T q + \lambda q  \\ 
        q^T \tilde{\A}_j^i + \lambda q^T & \lambda r 
    \end{bmatrix}
    \label{invariant set semi-definite program}
\end{equation}
is negative semi-definite for all $j=1,\dots,N$. for a fixed value of $\lambda$, looking for $\bar{\mathcal{P}}_i$ such that (\ref{invariant set semi-definite program}) is negative semi-definite is convex in $\bar{\mathcal{P}}_i$. In particular, setting $\lambda=0$ brings back the search for a polynomial Lyapunov function. Other values of $\lambda$ then offer other options for finding polynomial invariant sets. The authors in \cite{abdelraouf2023algebraic} developed a method to construct a Lyapunov function of degree 1 from the invariant set. This homogeneous Lyapunov function is called Algebraic Lyapunov function. 
For the system presented in example 1, the convex feasibility problem (\ref{invariant set semi-definite program}) is solved for $i=5 \text{ and } 7$ and for $\lambda=-0.05$ to find the invariant set passing through the initial condition $x_0=\left[ 1 \; \; 0\right]^T$. The results are shown in figure 4. 
\begin{figure}[H]
    \centering
    \includegraphics[scale=0.4]{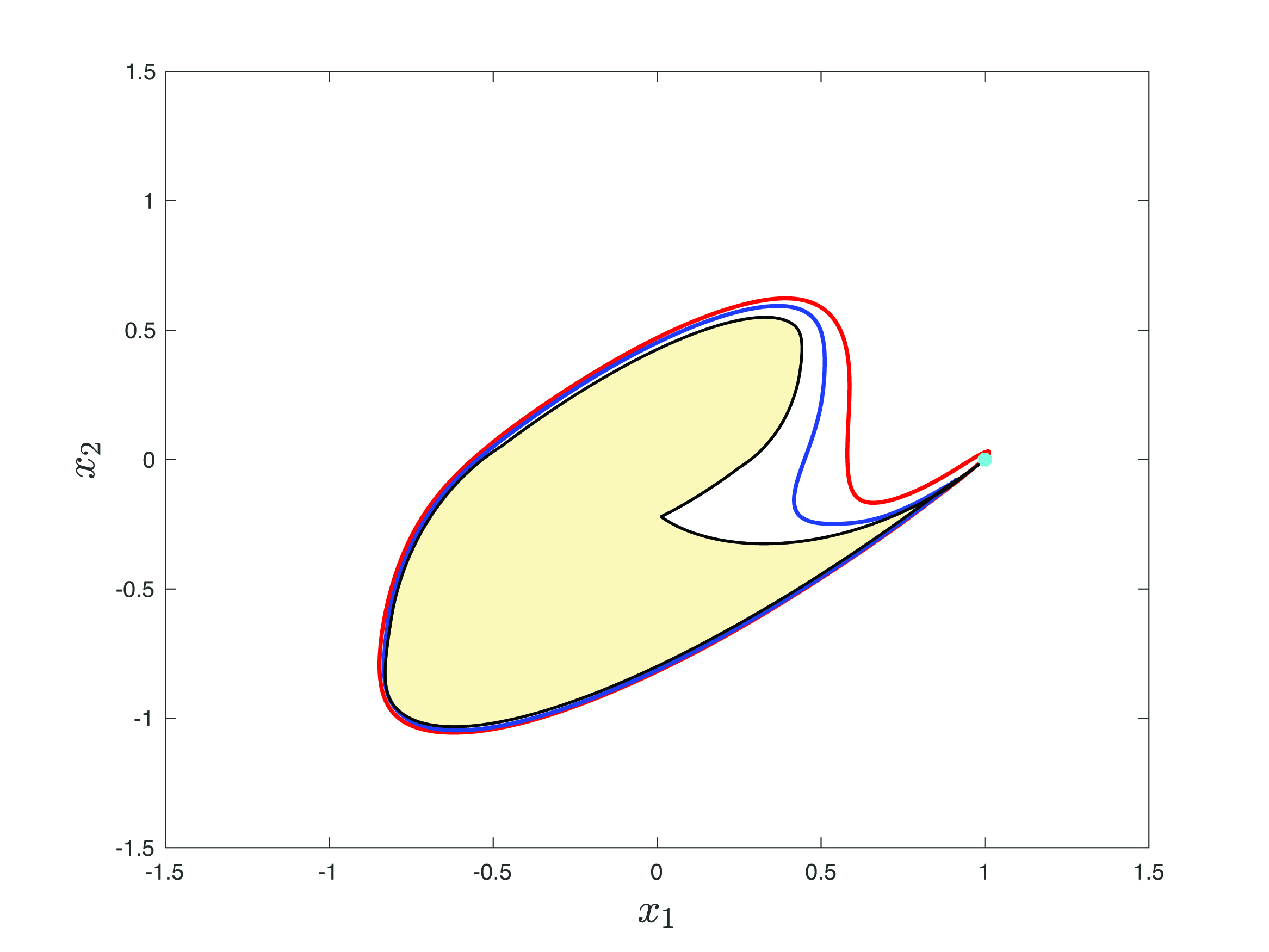}
    \caption{Simulated response for system presented in example 1. The light yellow region represents the set of reachable states from the initial condition $x_0$. The red and blue regions represents the non-homogeneous polynomial invariant sets resulted from solving (\ref{invariant set semi-definite program}) for $i=5 \text{ and } 7$ respectively and for $\lambda =-0.05$. }
    \label{fig:4}
\end{figure}

%%%%%%%%%%%%%%%%%%%%%%%%%%%%%%%%%%%%%%%%%%%%%%%%%%%%%%%%
\section{Relation to Polynomial Lyapunov functions obtained via "sum-of-squares" semidefinite relaxation methods}
The authors believe the foregoing quadratic Lyapunov functions are
in direct correspondence with the polynomial Lyapunov functions for System~(\ref{ LTV system}) that may be computed by means of "Sum-of-Square" relaxations, and we formulate the following conjecture:

\noindent {\bf Conjecture:} There exists a positive definite matrix ${\cal P}_i\in \R^{\tilde{n}_i \times \tilde{n}_i}$ that satisfies
\[
\mathcal{P}_i \tilde{\A}^i_j + (\tilde{\A}^i_j)^T \mathcal{P}_i \preceq 0 \; \; \text{for all } j =1,\dots,N
\]
if and only if there exists a polynomial $G(x)$ of degree $2i$ with $G(0)=0$ and $G(x)>0, x\neq 0$, satisfying 
\[
\frac{\partial G}{\partial x} . A_j x <0 \; \; \text{for all } j=1,\dots,N
\]
proving that conjecture is left for future work.

%%%%%%%%%%%%%%%%%%%%%%%%%%%%%%%%%%%%%%%%%%%%%%%%%%%%%%%

\section{Conclusion and future work}
This work presents a procedure to compute non-homogeneous polynomial Lyapunov functions and invariant sets for LTV systems. This procedure is based on building a hierarchy of LTV systems such that the quadratic Lyapunov function for any system in the hierarchy is a non-homogeneous polynomial Lyapunov function for the base level system. These non-homogeneous polynomials are shown to be effective in LTV systems performance analysis by computing outer approximation for the reachable sets, bounds for the impulse response and invariant sets that are not Lyapunov level sets. In future work,  the theoretical relation between the generated non-homogeneous polynomials and other polynomials computed by means of Sum-of-Square relaxations will be considered. 

\bibliographystyle{ieeetr}

\bibliography{main.bib}

\begin{thebibliography}{10}

\bibitem{boyd1994linear}
S.~Boyd, L.~El~Ghaoui, E.~Feron, and V.~Balakrishnan, {\em Linear matrix
  inequalities in system and control theory}.
\newblock SIAM, 1994.

\bibitem{ahmadi2017sum}
A.~A. Ahmadi and P.~A. Parrilo, ``Sum of squares certificates for stability of
  planar, homogeneous, and switched systems,'' {\em IEEE Transactions on
  Automatic Control}, vol.~62, no.~10, pp.~5269--5274, 2017.

\bibitem{mason2006common}
P.~Mason, U.~Boscain, and Y.~Chitour, ``Common polynomial {L}yapunov functions
  for linear switched systems,'' {\em SIAM journal on control and
  optimization}, vol.~45, no.~1, pp.~226--245, 2006.

\bibitem{liberzon2003switching}
D.~Liberzon, {\em Switching in systems and control}, vol.~190.
\newblock Springer, 2003.

\bibitem{zelentsovsky1994nonquadratic}
A.~Zelentsovsky, ``Nonquadratic {L}yapunov functions for robust stability
  analysis of linear uncertain systems,'' {\em IEEE Transactions on Automatic
  control}, vol.~39, no.~1, pp.~135--138, 1994.

\bibitem{abate2020lyapunov}
M.~Abate, C.~Klett, S.~Coogan, and E.~Feron, ``Lyapunov differential equation
  hierarchy and polynomial {L}yapunov functions for switched linear systems,''
  in {\em 2020 American Control Conference (ACC)}, pp.~5322--5327, IEEE, 2020.

\bibitem{parrilo2003semidefinite}
P.~A. Parrilo, ``Semidefinite programming relaxations for semialgebraic
  problems,'' {\em Mathematical programming}, vol.~96, pp.~293--320, 2003.

\bibitem{jarvis2005control}
Z.~Jarvis-Wloszek, R.~Feeley, W.~Tan, K.~Sun, and A.~Packard, ``Control
  applications of sum of squares programming,'' {\em Positive polynomials in
  control}, pp.~3--22, 2005.

\bibitem{parrilo2000structured}
P.~A. Parrilo, {\em Structured semidefinite programs and semialgebraic geometry
  methods in robustness and optimization}.
\newblock California Institute of Technology, 2000.

\bibitem{abedor1996linear}
J.~Abedor, K.~Nagpal, and K.~Poolla, ``A linear matrix inequality approach to
  peak-to-peak gain minimization,'' {\em International Journal of Robust and
  Nonlinear Control}, vol.~6, no.~9-10, pp.~899--927, 1996.

\bibitem{blanchini1995nonquadratic}
F.~Blanchini, ``Nonquadratic {L}yapunov functions for robust control,'' {\em
  Automatica}, vol.~31, no.~3, pp.~451--461, 1995.

\bibitem{abate2021pointwise}
M.~Abate, C.~Klett, S.~Coogan, and E.~Feron, ``Pointwise-in-time analysis and
  non-quadratic {L}yapunov functions for linear time-varying systems,'' in {\em
  2021 American Control Conference (ACC)}, pp.~3550--3555, IEEE, 2021.

\bibitem{abdelraouf2022computing}
H.~Abdelraouf, G.-Y. Immanuel, and E.~Feron, ``Computing bounds on $
  l_{\infty}$-induced norm for linear time-invariant systems using homogeneous
  {L}yapunov functions,'' {\em arXiv preprint arXiv:2203.00716}, 2022.

\bibitem{lin2009stability}
H.~Lin and P.~J. Antsaklis, ``Stability and stabilizability of switched linear
  systems: a survey of recent results,'' {\em IEEE Transactions on Automatic
  control}, vol.~54, no.~2, pp.~308--322, 2009.

\bibitem{chesi2009homogeneous}
G.~Chesi, A.~Garulli, A.~Tesi, and A.~Vicino, {\em Homogeneous polynomial forms
  for robustness analysis of uncertain systems}, vol.~390.
\newblock Springer Science \& Business Media, 2009.

\bibitem{vandenberghe1993polynomial}
L.~Vandenberghe and S.~Boyd, ``A polynomial-time algorithm for determining
  quadratic {L}yapunov functions for nonlinear systems,'' in {\em Eur. Conf.
  Circuit Th. and Design}, pp.~1065--1068, 1993.

\bibitem{grant2014cvx}
M.~Grant and S.~Boyd, ``Cvx: Matlab software for disciplined convex
  programming, version 2.1,'' 2014.

\bibitem{feronThesis}
E.~Feron, ``Linear matrix inequalities for the problem of absolute stability of
  control systems,'' {\em Ph.D. Thesis, Stanford University}, 1994.

\bibitem{klett2021numerical}
C.~Klett, M.~Abate, S.~Coogan, and E.~Feron, ``A numerical method to compute
  stability margins of switching linear systems,'' in {\em 2021 American
  Control Conference (ACC)}, pp.~864--869, IEEE, 2021.

\bibitem{abdelraouf2023algebraic}
H.~Abdelraouf, E.~Feron, and J.~Shamma, ``Algebraic {L}yapunov functions for
  homogeneous dynamic systems,'' {\em arXiv preprint arXiv:2303.02185}, 2023.

\end{thebibliography}
\appendix
 \subsection{Dimensionality reduction}
Using Kronecker product produces redundant states that significantly increase the dimension of the lifted system. For example, if $x \in \R^2$, then $  \xi_2  \in \R^{4}$ such that $ \xi_2 =  x \otimes x= \begin{bmatrix} x_1^2 & x_1 x_2 & x_1x_2 & x_2^2\end{bmatrix}^T$. To get rid of the redundant states, a simple procedure introduced in \cite{abate2020lyapunov}. This procedure is based on introducing a new variable $\eta_2 \in \R^3$ such that $\xi_2 =W_2 \eta_2$ where 
\[
W_2 = \begin{bmatrix}
     1 &0&0\\
    0 & 1& 0 \\
    0 & 1 & 0\\
    0 & 0 & 1
\end{bmatrix}.
\]
For second order systems, i.e. $n=2$, this method can be generalized for $k\geq 2 $ as $\xi_k=W_k \eta_k$ such that
\[
W_k = \left[\begin{array}{cc}
    W_{k-1} & 0_{2^{k-1}}  \\
     0_{2^{k-1}}& W_{k-1} 
\end{array}\right]
\]
 and $W_1=I_2$. So, the dimension of the system $H_k$ in the hierarchy (\ref{systems Hierarcy}) is reduced from $n^k$ to $k+1$. Without loss of generality, if $\mathcal{M}=\{A\}$, then every system $\dot{\xi}_k=\A^k \xi_k $ in the hierarchy (\ref{systems Hierarcy}) can be written as $\dot{\eta}_k=W_k^+ \A^k W_k \eta_k $ where $W_k^+$ is the pseudo-inverse of $W_k$ for $k\ge1$. We can then use the  reduced dimension state vector $\eta_k$ to rebuild the new hierarchy (\ref{ new systems Hierarcy}). Hence, $\tilde{\eta}_i=\left[\eta_1^T,\eta_2^T,\dots, \eta_i^T\right]^T$ which reduces the dimension of the system $\tilde{H}_i$ in the hierarchy (\ref{ new systems Hierarcy})  from $2(2^i-1)$ to $i(i+3)/2$. Thus, the reduced dimension system $\tilde{H}_i$  will be $\tilde{\eta}_i= \tilde{W}_i^+ \tilde{\A}^i \tilde{W}_i \tilde{\eta}$ where $\tilde{W}_i = \text{diag}(W_1,W_2,\dots,W_i)$.

\end{document}